\documentclass[11pt]{article}

\usepackage{xcolor, verbatim}

\usepackage[pagebackref]{hyperref}

\renewcommand{\backref}[1]{}

\renewcommand{\backrefalt}[4]{%
\ifcase #1
\or $^{#2}$%
\else $^{#2}$%
\fi}
\usepackage{kpfonts}
\usepackage[margin=1in]{geometry}
\hypersetup{
    colorlinks,
    linkcolor={blue!100!black},
    citecolor={blue!100!black},
}

\usepackage{fullpage}
\usepackage{amsmath,amsthm,amsfonts,dsfont}
\usepackage{amssymb,latexsym,graphicx}
\usepackage{palatino}
\usepackage{mathpazo}
\usepackage{stmaryrd}
\usepackage{mathtools}
\usepackage{hyperref}
\usepackage{xspace}
\usepackage{float}

\newtheorem{theorem}{Theorem}[section]

\newtheorem*{remark}{Remark}
\newtheorem{proposition}[theorem]{Proposition}
\newtheorem{lemma}[theorem]{Lemma}

\newtheorem{corollary}[theorem]{Corollary}
\newtheorem{definition}[theorem]{Definition}

  \makeatletter
  \newtheorem*{rep@theorem}{\rep@title}
  \newcommand{\newreptheorem}[2]{%
  \newenvironment{rep#1}[1]{%
  \def\rep@title{#2 \ref{##1}}%
  \begin{rep@theorem}}%
  {\end{rep@theorem}}}
  \makeatother
  
  \newreptheorem{lemma}{Lemma}
  \newreptheorem{theorem}{Theorem}
  \newreptheorem{corollary}{Corollary}
  \newreptheorem{proposition}{Proposition}
  \newreptheorem{conjecture}{Conjecture}

\usepackage{mathtools}

\DeclareMathOperator{\poly}{poly}
\newcommand{\C}{\ensuremath{\mathbb{C}}}
\newcommand{\N}{\ensuremath{\mathbb{N}}}

\newcommand{\R}{\ensuremath{\mathbb{R}}}

\newcommand{\mc}[1]{\ensuremath{\mathcal{#1}}\xspace}

\newcommand{\ket}[1]{|#1\rangle}

\newcommand{\id}{\mathbb{1}}

\newcommand{\Y}{\mathcal{Y}}

\newcommand{\X}{\mathcal{X}}

\newcommand{\HH}{\mathcal{H}}

\newcommand{\LH}{L(\mathcal H)}

\newcommand{\KK}{\mathcal{K}}

\newcommand{\eps}{\varepsilon} 
\renewcommand{\epsilon}{\varepsilon} 
\DeclareMathOperator{\cb}{cb}

\DeclareMathOperator{\im}{Im}

\newcommand{\ldc}{L(\C^d)}

\newcommand{\cbdeg}{\text{cb-deg}}
\newcommand{\bmdeg}{\text{bm-deg}}

\renewcommand{\epsilon}{\varepsilon}

\newcommand{\pmset}[1]{\{-1,1\}^{#1}} % hypercube in +-1 basis
\newcommand{\bset}[1]{\{0,1\}^{#1}} % hypercube
\newcommand{\st}{:\,} % "such that" to define sets
\newcommand{\ie}{{i.e.}}

\renewcommand{\Pr}{\mbox{\rm Pr}}
\DeclareMathOperator{\diag}{diag}
\DeclareMathOperator{\Diag}{Diag}
\newcommand{\Exp}{\mathbf{E}}

\DeclareMathOperator{\Span}{Span}

\newcommand{\beq}{\begin{equation}}
\newcommand{\eeq}{\end{equation}}
\newcommand{\beqn}{\begin{equation*}}
\newcommand{\eeqn}{\end{equation*}}
\newcommand{\beqr}{\begin{eqnarray}}
\newcommand{\eeqr}{\end{eqnarray}}
\newcommand{\beqrn}{\begin{eqnarray*}}
\newcommand{\eeqrn}{\end{eqnarray*}}
\newcommand{\bmline}{\begin{multline}}
\newcommand{\emline}{\end{multline}}
\newcommand{\bmlinen}{\begin{multline*}}
\newcommand{\emlinen}{\end{multline*}}

  \floatstyle{ruled}
  \newfloat{algorithm}{h}{lob}[section]
  \floatname{algorithm}{Algorithm}

\begin{document}
\title{Quantum Query Algorithms are Completely Bounded Forms}
\author{
Srinivasan Arunachalam\thanks{QuSoft, CWI and University of Amsterdam, the Netherlands. Supported by ERC Consolidator Grant QPROGRESS. E-mail: \texttt{arunacha@cwi.nl}}
\and
Jop Bri\"et\thanks{CWI, QuSoft. 
%Science Park 123, 1098 XG Amsterdam, Netherlands.
Supported by a VENI grant and the Gravitation-grant NETWORKS-024.002.003 from the Netherlands Organisation for Scientific Research~(NWO). E-mail: \texttt{j.briet@cwi.nl}}\\
\and
Carlos Palazuelos\thanks{Facultad de C.C. Matematicas, UCM. Instituto de Ciencias Matematicas, Madrid Spain. 
Supported by the Ramon y Cajal program (RYC-2012-10449), the Spanish MINECO MTM2014-54240-P, Comunidad de Madrid (QUITEMAD+ Project S2013/ICE-2801) and ICMAT Severo Ochoa Grant No. SEV-2015-0554. E-mail: \texttt{carlospalazuelos@mat.ucm.es}}
}
\date{}
\maketitle

\begin{abstract}
 We prove a characterization of $t$-query quantum algorithms in terms of the unit ball of a space of degree-$2t$ polynomials. 
Based on this, we obtain a refined notion of approximate polynomial degree that equals the quantum query complexity, answering a question of Aaronson et al.\ (CCC'16). 
Our proof is based on a fundamental result of Christensen and Sinclair~(\emph{J.\ Funct.\ Anal.}, 1987) that generalizes the well-known Stinespring representation for quantum channels to multilinear forms.
Using our characterization, we show that many polynomials of degree four are far from those coming from two-query quantum algorithms.
We also give a simple and short proof of one of the results of Aaronson et al.\ showing an equivalence between one-query quantum algorithms and bounded quadratic polynomials.

\emph{Revision note: A mistake was found in the proof of the second result on degree-4 polynomials far from 2-query quantum algorithms.
An explanation of the issue, a corrected proof and stronger examples are presented in work of Escudero Guti\'errez and the second author.}
\end{abstract}

% REQUIRED

\section{Introduction}

In the black-box model of quantum computation one is given access to a unitary operation, usually referred to as an oracle, that allows one to probe the bits of an unknown binary string ${x\in\pmset{n}}$ in superposition. Promised that~$x$ lies in a subset $D\subseteq \pmset{n}$, the goal in this model is to learn some property of $x$ given by a Boolean function $f:D\to \pmset{}$, when only given access to $x$ through the oracle. An application of the oracle is usually referred to as a \emph{query}.
The bounded-error quantum query complexity of~$f$, denoted~$Q_{\eps}(f)$, is the minimal number of queries a quantum algorithm must make on the worst-case input $x\in D$ to compute~$f(x)$ with probability at least~$1-\eps$, where~$\eps \in (0,1/2)$ is usually some fixed but arbitrary positive constant.

Many of the best-known quantum algorithms are naturally captured by this model.
A few examples of partial functions whose quantum query complexity is exponentially smaller than their classical counterpart (the decision-tree complexity) are period finding~\cite{Shor:1997}, Simon's problem~\cite{simon:power} and Forrelation~\cite{Aaronson:2015}. Famous problems related to total functions that admit polynomial quantum speed-ups include unstructured search~\cite{Grover:1996}, element distinctness~\cite{ambainis:qwalk} and NAND-tree evaluation~\cite{Farhi:2008}. 
It is well-known that for all total functions, the quantum and classical query complexities are polynomially related~\cite{Beals:2001}; see Ambainis et al.~\cite{ambainis:pointerfunction} and Aaronson et al.~\cite{aaronson:cheatsheet} for recent progress on the largest possible separations.

Despite the simplicity of the query model, determining the quantum query complexity of a given function $f$ appears to be highly non-trivial.
Several methods were introduced to tackle this problem. For constructing quantum query algorithms, there are general methods based on quantum walks~\cite{ambainis:qwalk,mnrs:quantumwalk}, span programs~\cite{reichardt:span} and learning graphs~\cite{belovs:learninggraphs}.
For proving lower bounds there are two main methods, known as the \emph{polynomial method}~\cite{Beals:2001} and the \emph{adversary method}~\cite{ambainis:positiveadv}.
The latter was eventually generalized to the ``negative weight'' adversary method~\cite{Hoyerleespalek:negative} and was shown to \emph{characterize} quantum query complexity~\cite{Hoyerleespalek:negative,reichardt:span,reichardt:querycompose,Lee:quantumstateconversion}, but proving  lower bounds using this method appears to be hard in general. 
This paper will focus on the polynomial~method.

\subsection{The polynomial method} 
The polynomial method is based on a connection between quantum query algorithms and polynomials discovered by Beals et al.~\cite{Beals:2001}.
They observed that for every $t$-query quantum algorithm~$\mathcal A$ that on input $x\in\pmset{n}$ returns a sign~$\mathcal A(x)$, there exists a degree-$(2t)$ polynomial~$p$ such that $p(x) = \Exp[\mathcal A(x)]$ for every~$x$, where the expectation is  over the randomness of the output (note that this is the difference of the acceptance and rejection probabilities of the algorithm). Let $D\subseteq \pmset{n}$ and $f:D\rightarrow \pmset{}$ be a (possibly partial) Boolean function. From the observation it follows that if $\mathcal A$ computes $f$ with probability at least $1-\eps$, then $p$ satisfies $|p(x) - f(x)| \leq  2\eps$ for every $x\in D$. The polynomial method thus converts the problem of lower bounding quantum query complexity to the problem of proving lower bounds on the minimum degree of a polynomial~$p$ such that $|p(x) - f(x)| \leq  2\eps$ holds for inputs~$x\in D$. The minimal degree of such a polynomial is called the \emph{approximate (polynomial) degree} and is denoted by~$\deg_{\eps}(f)$.

Notable applications of this approach showed optimality for Grover's search algorithm~\cite{Beals:2001}\footnote{The first quantum lower bound for the search problem was proven by Bennett et al.~\cite{bbbv:str&weak} using the so-called hybrid method. Beals et al.~\cite{Beals:2001} reproved their result using the polynomial method.}
and the above-mentioned algorithms for collision-finding and element distinctness~\cite{Aaronson:2004}.  In a recent work, Bun et al.~\cite{bundualpoly:2017} use the polynomial method to resolve the quantum query complexity of several other well-studied Boolean functions.
\medskip

\paragraph{Converses to the polynomial method}
A natural question is whether the polynomial method admits a converse. 
If so, this would imply a succinct characterization of quantum algorithms in terms of basic mathematical objects. 
However, Ambainis~\cite{Ambainis:2006} answered this question in the negative, showing that for infinitely many~$n$,  there is a total function~$f$ with~$\deg_{1/3}(f) \leq n^\alpha$ and $Q_{1/3}(f)\geq  n^{\beta}$ for some positive constants $\beta > \alpha$ (recently larger separations were obtained for total functions by Aaronson et al.~\cite{aaronson:cheatsheet}).\footnote{An open problem of Aaronson~\cite{aaronson:slides} asks whether for partial Boolean functions there exists an \emph{exponential} separation between $\deg_\eps(f)$ and $Q_\eps(f)$.}
The approximate degree thus turns out to be an imprecise measure for quantum query complexity in general.
These negative results would still leave room for the following two~possibilities:
\vspace{2 pt}
\begin{enumerate}
	\item There is a (simple) refinement of approximate polynomial degree that approximates $Q_\eps(f)$ up to a constant factor.
	\item Constant-degree polynomials characterize constant-query quantum algorithms.
\end{enumerate}
\vspace{2 pt}
These avenues were recently explored by Aaronson et al.~\cite{Aaronson:2015, Aaronson:2016}. 
The first work strengthened the polynomial method by observing that quantum algorithms give rise to polynomials with a so-called \emph{block-multilinear} structure.
Based on this observation, they introduced a refined degree measure, $\bmdeg_{\eps}(f)$ which lies between $\deg_\eps(f)$ and $2Q_\eps(f)$, prompting the immediate question of how well that approximates~$Q_\eps(f)$. 
The subsequent work  showed, among other things, that for infinitely many $n$, there is a function~$f$ with $\bmdeg_{1/3}(f)=O(\sqrt{n})$ and $Q_{1/3}(f)= \Omega(n)$, thereby also ruling out the possibility that this degree measure validates possibility~1. 
The natural next question then asks if there is another refined notion of polynomial degree that approximates quantum query complexity~\cite[Open problem~3]{Aaronson:2016}.

In the direction of the second avenue, \cite{Aaronson:2016} showed a surprising  converse to the polynomial method for quadratic polynomials.
Say that a polynomial $p\in \R[x_1,\dots,x_n]$ is \emph{bounded} if it satisfies $p(x)\in [-1,1]$ for all $x\in \pmset{n}$.

\begin{theorem}[Aaronson et al.]\label{thm:AAIKS}
	There exists an absolute constant~$C\in (0,1]$ such that the following holds.
	For every bounded quadratic polynomial~$p$, there exists a one-query quantum algorithm that, on input~$x\in\pmset{n}$, returns a sign with expectation $Cp(x)$.
\end{theorem}

This implies that possibility~2 holds true for quadratic polynomials.
It also leads to the problem of finding a similar converse for higher-degree polynomials, asking for instance whether two-query quantum algorithms are equivalent to quartic polynomials~\cite[Open problem~1]{Aaronson:2016}.

\subsection{Our results}

This paper addresses the above-mentioned two problems.
Our first result is a new notion of  polynomial degree
that gives a tight characterization of quantum query complexity (Definition~\ref{def:cbdeg} and Corollary~\ref{cor:cbdeg} below), giving an answer to~\cite[Open problem~3]{Aaronson:2016}. 
Using this characterization, we show that there is no generalization of Theorem~\ref{thm:AAIKS} to higher-degree polynomials, in the sense that there is no absolute constant $C\in (0,1]$ for which the analogous statement holds true.
This gives a partial answer to~\cite[Open problem 1]{Aaronson:2016}, ruling out a strong kind of equivalence.
Finally, we give a simplified shorter proof of Theorem~\ref{thm:AAIKS}.
Below we explain our results in more~detail. 
\medskip

\paragraph{Quantum algorithms are completely bounded forms}
For the rest of the discussion, all polynomials will be assumed to be bounded, real and $(2n)$-variate if not specified otherwise. We refer to a homogeneous polynomial as a \emph{form}.
For $\alpha\in \{0,1,2,\dots\}^{2n}$ and $x\in \R^{2n}$, we write $|\alpha| = \alpha_1 + \cdots + \alpha_{2n}$ and $x^\alpha = x_1^{\alpha_1}\cdots x_{2n}^{\alpha_{2n}}$.
Then, any form~$p$ of degree~$t$ can be written as 
\beq\label{eq:form}
p(x) =\sum_{\alpha\in \{0,1,\dots,t\}^{2n}:\, |\alpha| = t} c_\alpha x^\alpha,
\eeq
where~$c_\alpha$ are some real coefficients.
Our new notion of polynomial degree is based on a characterization of quantum query algorithms in terms of forms satisfying a certain norm constraint.
The norm we assign to a form as in~\eqref{eq:form} is given by a norm of the unique symmetric $t$-tensor $T_p\in\R^{{2n}\times\cdots\times {2n}}$ such that~$p$ can be written as 
\beq\label{eq:pTdecomp}
p(x)
=
\sum_{i_1,\dots,i_t=1}^{2n} (T_p)_{i_1,\dots,i_t}x_{i_1}\cdots x_{i_t}.
\eeq 
Explicitly, this tensor is given by
\beq\label{eq:Tpdef}
(T_p)_{i_1,\dots,i_t} = \frac{c_{e_{i_1} + \cdots + e_{i_t}}}{\tau(i_1,\dots,i_t)},
\eeq 
where $e_i$ is the $i$th standard basis vector for~$\R^{2n}$ and $\tau(i_1,\dots,i_t)$ is the number of distinct permutations of the sequence $(i_1,\dots,i_t)$.
The relevant norm of~$T_p$ is in turn given in terms of an infimum over decompositions of the form $T_p = \sum_{\sigma\in S_t} T^{\sigma}\circ\sigma,$ where the sum is over permutations of~$\{1,\dots,t\}$, each~$T^{\sigma}$ is a $t$-tensor, and $T^\sigma\circ \sigma$ is the permuted version of~$T^{\sigma}$ given by
$$
(T^{\sigma}\circ \sigma)_{i_1,\dots,i_t} = T^{\sigma}_{i_{\sigma(1)},\dots, i_{\sigma(t)}}.
$$
Note that the notation $T^{\sigma}$ does not refer to an action of~$S_t$ on the set of tensors.
Moreover, since~$T^\sigma$ is arbitrary we could have just absorbed the permutation  in the decomposition of~$T_p$; the reason why we didn't will become clear in a moment.
Finally, the actual norm is based on the \emph{completely bounded norm} of each of the~$T^{\sigma}$. 
Given a $t$-tensor~$T\in\R^{{2n}\times \cdots\times {2n}}$, its completely bounded norm~$\|T\|_{\cb}$ is given by the supremum over positive integers~$k$ and collections of $k\times k$ unitary matrices $U_1(i),\dots,U_t(i)$, for $i\in[{2n}]$, of the operator norm
\beq \label{eq:defnTcb}
\Big\|\sum_{i_1,\dots,i_{t}=1}^{2n}T_{i_1,\dots,i_{t}} U_1(i_1)\cdots U_{t}(i_{t})\Big\|.
\eeq

\begin{definition}[Completely bounded norm of a form]\label{defn:pcb}
	Let~$p$ be a form of degree~$t$ and let~$T_p$ be the symmetric $t$-tensor as in~\eqref{eq:Tpdef}.
	Then, the \emph{completely bounded norm} of~$p$ is defined by
	\beq \label{eq:defnpcb}
	\|p\|_{\cb}
	=
	\inf\Big\{\:\sum_{\sigma\in S_t}\|T^{\sigma}\|_{\cb}
	\:\st\:
	T_p = \sum_{\sigma\in S_t}T^{\sigma}\circ\sigma\:
	\Big\}.
	\eeq
\end{definition}

Standard compactness arguments show that both the completely bounded norm of tensors and of polynomials are attained.
Let us point out that~$\|T\|_{\cb}$ does not always equal $\|T\circ\sigma\|_{\cb}$ for a non-trivial permutation. 
For this reason, the completely bounded norm of a polynomial can be significantly smaller than that of its associated symmetric tensor: 
for $n$-variate cubic forms their ratio can be as large as~$\Omega(\sqrt{n})$.
Let us also mention that for ease of exposition, we are abusing the term ``completely bounded norm''.
Such norms originate from operator space theory and make sense only in reference to underlying operator spaces, which we have tacitly fixed in the above discussion.
The norm in~\eqref{eq:defnpcb} was originally introduced in the general context of tensor products of operator spaces in \cite{OiPi:1999}. In that framework, the definition considered here corresponds to a particular operator space based on~$\ell_1^n$, but we shall not use this fact here. 

Our characterization of quantum query algorithms is as follows.

\begin{theorem}[Characterization of quantum algorithms]\label{thm:main}
	Let $\beta:\pmset{n}\to [-1,1]$ and let~$t$ be a positive integer.
	Then, the following are equivalent.
	\begin{enumerate}
		\item There exists a form~$p$ of degree~$2t$ such that $\|p\|_{\cb} \leq 1$ and $p((x, {\bf 1})) = \beta(x)$ for every $x\in \pmset{n}$, where ${\bf 1}\in\R^n$ is the all-ones vector.\footnote{In a follow-up work, Gribling and Laurent~\cite{GL:2019} observed that ${\bf 1}\in \R^n$ can in fact be replaced by a single $1$.}
		\item There exists a $t$-query quantum algorithm that, on input $x\in\pmset{n}$, returns a  sign with expected value~$\beta(x)$.
	\end{enumerate}
\end{theorem}

It may be observed that  the polynomial method is contained in the above statement, since any $(2n)$-variate form~$p$ defines an $n$-variate polynomial given by $q(x) = p((x,{\bf 1}))$.
The above theorem refines the polynomial method in the sense that quantum algorithms can only yield polynomials of the form $q(x) = p((x, {\bf 1}))$ where~$p$ has completely bounded norm at most one.

Our proof is based on a fundamental result of Christensen and Sinclair~\cite{christensen:cbrep} concerning multilinear forms on~$C^*$-algebras that generalizes the well-known Stinespring representation theorem for quantum channels (see also~\cite{paulsenandsmith:multilinear} and~\cite[Chapter~5]{Pisier:2003}).
As such, this result applies in a more general setting than what is strictly needed here.
Section~\ref{sec:prelims} contains some preliminary material that will allow us to state the result in its original form, in particular the general definition of completely bounded norms of multilinear forms on~$C^*$-algebras.
\medskip

\paragraph{Completely bounded approximate degree}
Theorem~\ref{thm:main} motivates the following new notion of approximate degree for partial Boolean functions.

\begin{definition}[Completely bounded approximate degree]\label{def:cbdeg} 
	For $D\subseteq \pmset{n}$, let $f:D\rightarrow \pmset{}$ be a (possibly partial) Boolean function and let $\eps\geq 0$.
	Then, the $\eps$-\emph{completely bounded approximate degree} of~$f$, denoted $\cbdeg_\eps(f)$, is the smallest positive integer~$t$ for which there exists a form $p$ of degree~$2t$ such that
	$\|p\|_{\cb} \leq 1$ as in Eq.~\eqref{eq:defnpcb} and we have $|p((x, {\bf 1}))-f(x)|\leq 2\varepsilon$ for every $x\in D$.
\end{definition}
As a corollary of Theorem~\ref{thm:main}, we get the following characterization of quantum query complexity.

\begin{corollary}\label{cor:cbdeg}
	For every $D\subseteq \pmset{n}$, $f:D\rightarrow \pmset{}$ and $\eps\geq 0$, we have
	$\cbdeg_\eps(f)=Q_\eps(f)$.
\end{corollary}

We remark that the characterization of $Q_\eps(f)$ via the adversary method holds for all constant $\eps>0$, whereas our characterization  holds for every $\eps \geq 0$. In addition, in our characterization we do not lose constant factors (unlike in the adversary method characterization) which could possibly be useful to understand the quantum query complexity of ordered search~\cite{hoyer:orderedsearch,childs:orderedsearch}.
\vspace{2 pt}
\paragraph{Chebyshev polynomials}
The Chebyshev polynomials have been used in a number of places to find approximating polynomials for Boolean functions, most notably~\cite{nisan&szegedy:degree}.
These polynomials can be defined through the recursion $T_0(\alpha) = 1, T_1(\alpha) =\alpha, T_{k+1}(\alpha) = 2\alpha T_{k}(\alpha) - T_{k-1}(\alpha)$ for $k\in \N$.
Particularly useful are the $n$-variate degree-$k$ polynomials $p_k(x) = T_k\big((x_1 + \cdots + x_n)/n\big)$.
In a forthcoming work, we show using a straightforward argument based on the recursion formula that there exist degree-$k$ forms~$F_k$ on $\R^n$ such that $F_k(x)=p_k(x)$ for every $x\in \pmset{n}$ and $\|F_k\|_{\cb} \leq 1$ for every~$k$.
As a simple application, from Theorem~\ref{thm:main} and a result of~\cite{nisan&szegedy:degree}, one then easily obtains the fact that the $n$-bit {\sc OR} function, restricted to the set of strings with Hamming weight at most~1, has quantum query complexity~$O(\sqrt{n})$, as implied by Grover's algorithm~\cite{Grover:1996}. 

\vspace{2 pt}

\paragraph{Separations for higher-degree forms}
Theorem~\ref{thm:AAIKS} follows from our Theorem~\ref{thm:main} and the fact that for every bounded quadratic form $p(x) = x^TAx$, the matrix $A$  has completely bounded norm bounded from above by an absolute constant (independent of~$n$); this is discussed in more detail below. 
If the same were true for the tensors~$T_p$ corresponding to higher-degree forms~$p$ then Theorem~\ref{thm:main} would give higher-degree extensions of Theorem~\ref{thm:AAIKS}.
Unfortunately, this will turn out to be false for polynomials of degrees  greater than $3$. Bounded forms whose associated tensors have unbounded completely bounded norm appeared before in the work of Smith~\cite{smith1988completely}, who gave an explicit example with completely bounded norm~$\sqrt{\log n}$.
Since~$\|p\|_{\cb}$ involves an infimum over decompositions of~$T_p$, this does not yet imply a counterexample to higher-degree versions of Theorem~\ref{thm:AAIKS}.
However, such counterexamples are implied by recent work on Bell inequalities, multiplayer XOR games in particular.
It is not difficult to see that~$\|p\|_{\cb}$ is bounded from below by the so-called \emph{jointly completely bounded norm} of the tensor~$T_p$, a quantity that in quantum information theory is better known as the entangled bias of the XOR game whose (unnormalized) game tensor is given by~$T_p$.
One obtains this quantity by inserting tensor products between the unitaries appearing in~\eqref{eq:defnTcb}.
P\'erez-Garc\'{i}a et al.~\cite{Perez-Garcia:2008} and Vidick and the second author~\cite{briet:explicitlowerbound} gave examples of bounded cubic forms with unbounded jointly completely bounded norm.
Both constructions are non-explicit, the first giving a completely bounded norm of order~$\Omega((\log n)^{1/4})$ and the latter of order~$\widetilde \Omega(n^{1/4})$. 
Here, we explain how to get a larger separation by means of a much simpler (although still non-explicit) construction and show that a  bounded cubic form $p$ given by a suitably normalized random sign tensor has completely bounded norm~$\|p\|_{\cb}=\Omega(\sqrt{n})$ with high probability (Theorem~\ref{thm:rand_T}). The result presented here is not new, but it follows from the existence of commutative operator algebras which are not $Q$-algebras. Here, we present a self-contained proof which follows the same lines as in \cite[Theorem 18.16]{DJT:1995} and, in addition, we prove the result with high probability (rather than just the existence of such trilinear forms). 
We also explain how to obtain from this result quartic examples by embedding into 3-dimensional ``tensor slices'', which in turn imply counterexamples to a quartic versus two-query version of Theorem~\ref{thm:AAIKS}.

\begin{remark}
A mistake was found in the last step mentioned above, concerning the implication of counterexamples to a quartic version of Theorem~\ref{thm:AAIKS}.
A corrected proof (and stronger examples) can be found in~\cite{BEG:2022}.
\end{remark}
\medskip

\paragraph{Short proof of Theorem~\ref{thm:AAIKS}}
As shown in~\cite{Aaronson:2016}, Theorem~\ref{thm:AAIKS} is yet another surprising consequence of the ubiquitous Grothendieck inequality~\cite{Grothendieck:1953} (Theorem~\ref{thm:grothendieck} below), well known for its relevance to Bell inequalities~\cite{Tsirelson:1985, Cleve:2004} and combinatorial optimization~\cite{Alon:2006, Khot:2012}, not to mention its fundamental importance to Banach spaces~\cite{Pisier:2012}. 
An equivalent formulation of Grothendieck's inequality again recovers Theorem~\ref{thm:AAIKS} for quadratic forms $p(x) = x^{\mathsf T}Ax$ given by a matrix $A\in \R^{n\times n}$ satisfying a certain norm constraint $\|A\|_{\ell_\infty\to\ell_1} \leq 1$, which in particular implies that~$p$ is bounded (see Section~\ref{sec:prelims} for more on this norm).
Indeed, in that case Grothendieck's inequality implies that $\|A\|_{\cb} \leq K_G$ for some absolute constant $K_G\in (1,2)$ (independent of~$n$ and~$A$).
Normalizing by~$K_G^{-1}$, one obtains Theorem~\ref{thm:AAIKS} with $C = K_G^{-1}$ for such quadratic forms from Theorem~\ref{thm:main}.
The general version of Theorem~\ref{thm:AAIKS} for quadratic polynomials follows from this via a so-called decoupling argument (see Section~\ref{sec:AAIKS-proof}).
This arguably does not simplify the original proof of Theorem~\ref{thm:AAIKS}, as Theorem~\ref{thm:main} relies on deep results itself.
However, in Section~\ref{sec:AAIKS-proof} we give a short simplified proof, showing that Theorem~\ref{thm:AAIKS} follows almost directly from a ``factorization version'' of Grothendieck's inequality (Theorem~\ref{thm:Grothfact}) that follows from the more standard version (Theorem~\ref{thm:grothendieck}).
The factorization version was used in the original proof as well, but only as a lemma in a more intricate argument.
In computer science, this factorization version has already found applications in an algorithmic version of the Bourgain--Tzafriri Column Subset Theorem~\cite{Tropp:2009} and algorithms for community detection in the stochastic block model~\cite{Le:2015}. This appears to be its first occurrence in quantum~computing.

\subsection{Related work}
\label{sec:related}
Although there is no converse to the polynomial method for arbitrary polynomials, equivalences between quantum algorithms and polynomials have been studied before in certain models of computation. For example, we do know of such characterizations in the model of non-deterministic query complexity~\cite{wolf:ndetq}, unbounded-error query complexity~\cite{buhrman:smallbias, montanaro:unboundederror} and quantum query complexity in expectation~\cite{kaniewski:querycomplexityinexp}. We remark here that in all these settings, the quantum algorithms constructed from polynomials were \emph{non-adaptive} algorithms, i.e., the quantum algorithm begins with a quantum state, repeatedly applies the oracle some fixed number of times and then performs a projective measurement. Crucially, these algorithms do not contain interlacing unitaries that are present in the standard model of query complexity, hence are known to be a much weaker class of algorithms (see Montanaro~\cite{montanaro:nonadaptive} for more details). 

Our main result is yet another demonstration of the expressive power of $C^*$-algebras and operator space theory in quantum information theory; for a survey on applications of these areas to two-prover one-round games, see~\cite{PaVi:2016}.
The appearance of $Q$-algebras (mentioned in the above paragraph on separations) is also not a first in quantum information theory, see for instance~\cite{Perez-Garcia:2008, Briet:2012, Briet:2013}.

After the initial version of this work appeared it was shown by Gribling and Laurent that the completely bounded norm of a degree-$d$ polynomial can be computed by a semidefinite program (SDP) of size~$O(n^d)$~\cite{GL:2019}. An SDP formulation for quantum query complexity was already known using the negative-weight adversary method~\cite{reichardt:querycompose}, but as we mentioned after Corollary~\ref{cor:cbdeg}, the adversary method only characterizes 
bounded-error quantum query complexity. 
With our 
characterization, the result of Gribling and Laurent gives a hierarchy 
of SDPs even for exact quantum query complexity.
An SDP characterization of quantum query complexity was also given earlier by Barnum, Saks and Szegedy~\cite{barnum:quantumquerysdp}. 
This SDP uses matrix-variables of size~$|\mathcal D|$, which is~$2^n$ for total functions, and so can be much larger than that of Gribling and Laurent.

\subsection{Organization}
In Section~\ref{sec:prelims}, we give a brief introduction to normed vector spaces, $C^*$-algebras and define the model of quantum query complexity. In Section~\ref{sec:mainthmproof}, we prove our main theorem characterizing quantum query algorithms. In Section~\ref{sec:cubicseparations}, we explain the separation obtained for higher-degree forms. In Section~\ref{sec:AAIKS-proof}, we give a short proof of the main theorem in Aaronson et al.~\cite{Aaronson:2016}.

\section{Preliminaries}

Here we fix some basic notation and recall some basic definitions.
In addition, in order to be able to state and use our main tool (Theorem~\ref{thm:SC-factor} of Christensen and Sinclair), we recall some basic facts of~$C^*$-algebras and completely bounded norms.

\label{sec:prelims}
\paragraph{Notation} 
For a positive integer $t$ denote $[t] = \{1,\dots,t\}$. 
For $x\in \C^n$, let $\Diag(x)$ be the $n\times n$ diagonal matrix whose diagonal is~$x$.
Given a matrix $X\in \C^{n\times n}$, let $\diag(X)\in \C^n$ denote the vector corresponding to the diagonal of $X$.
For $x\in \bset{n}$, denote $(-1)^x = ((-1)^{x_1},\dots,(-1)^{x_n})$. Let $e_1, \ldots, e_{n} \in \C^n$ be the standard basis vectors and let $E_{ij}=e_ie_j^*$. For $i,j\in [n]$, let $\delta_{i,j}$ be the indicator for the event $[i=j]$.
Let~${\bf 1} = (1,\dots,1)$ and ${\bf 0} = (0,\dots,0)$ denote the $n$-dimensional all-ones (resp.~all-zeros)~vector.

\paragraph{Normed vector spaces} 
For parameter $p\in [1,\infty)$, the $p$-norm of a vector $x\in \R^n$ is defined by $\|x\|_{\ell_p} = (|x_1|^p + \cdots + |x_n|^p)^{1/p}$ and for $p = \infty$ by $\|x\|_{\ell_\infty}=\max\{|x_i|\st i\in[n]\}$.
Denote the $n$-dimensional Euclidean unit ball by $B_2^n = \{x\in \R^n\st \|x\|_{\ell_2}\leq 1\}$.
For a matrix $A\in \R^{n\times n}$, denote the standard operator norm by~$\|A\|$ and define
$$
\|A\|_{\ell_\infty\to\ell_1} = \sup\big\{\|Ax\|_{\ell_1} \st \|x\|_{\ell_\infty} \leq 1\big\}
.$$ 
By linear programming duality, observe that the right-hand side of equality above
can be written as
$$
\sup\big\{\|Ax\|_{\ell_1} \st \|x\|_{\ell_\infty} \leq 1\big\}= \sup_{x,y\in \pmset{n}}x^{\mathsf T}Ay.
$$
We denote the norm of a general normed vector space~$X$ by $\|\cdot\|_X$, if there is a danger of ambiguity.
Denote by $\id_X$ the identity map on~$X$ and by $\id_d$ the identity map on~$\C^d$. For normed vector spaces $X,Y$, let $L(X,Y)$ be the collection of all linear maps $T:X\rightarrow Y$. We will use the notation $L(X)$ as a shorthand for $L(X,X)$. 
The (operator) norm of a linear map $T\in L(X,Y)$ is given by $\|T\| = \sup\{\|T(x)\|_Y\st \|x\|_X\leq 1\}$.
Such a map is an \emph{isometry} if $\|T(x)\|_Y = \|x\|_X$ for every $x\in X$ and a \emph{contraction} if $\|T(x)\|_Y \leq \|x\|_X$ for every  $x\in X$.
Throughout we endow~$\C^d$ with the standard Euclidean norm. 
Note that the space~$L(\C^d)$ is naturally identified with the set of~$d\times d$ matrices, sometimes denoted~$M_d(\C)$, and we use the two notations interchangeably.
For  Hilbert spaces~$\HH,\KK$, we endow $\HH\otimes \KK$ with the norm given by the inner product $\langle f\otimes a,g\otimes b\rangle = \langle f,g\rangle_\HH\langle a,b\rangle_\KK$, making this space isometric to $\HH\oplus\cdots\oplus \HH$ ($d$ times). 
This can be extended linearly to the entire domain. 
Similarly, we endow $\LH\otimes L(\C^d)$ with the operator norm of the space $L(\HH\otimes \C^d)$ of linear operators on the Hilbert space $\HH\otimes \C^d$; with some abuse of notation, we shall identify the two spaces of~operators. 
\medskip

\paragraph{$C^*$-algebras}
We collect a few basic facts of~$C^*$-algebras that we use later and refer to~\cite{Arveson:2012} for an extensive introduction.
A $C^*$-algebra $\mathcal X = (X,\cdot, *)$ is a normed complex vector space $X$, complete with respect to its norm (\ie, a Banach space), that is endowed with two operations in addition to the standard vector-space addition and scalar multiplication operations:
\begin{enumerate}
	\item an associative  multiplication $\cdot:X\times X\to X$, denoted $x\cdot y$ for $x,y\in X$, that is distributive with respect to the vector space addition and continuous with respect to the norm of~$X$, which by definition of continuity means $\|x\cdot y\|_X \leq \|x\|_X\|y\|_X$ for all $x,y\in X$;
	\item an involution $*: X\to X$, that is, a conjugate linear map that sends $x\in X$ to (a unique) $x^*\in X$ satisfying $(x^*)^* = x$ and $(xy)^* = y^*x^*$ for any $x,y\in X$, and such that $\|x\cdot x^*\|_X  = \|x\|_X^2$.
\end{enumerate}
Any finite-dimensional normed vector space is a Banach space. 
A $C^*$-algebra $\mc X$ is \emph{unital} if it has a multiplicative identity, denoted $\id_{\mc X}$.
The most important example of a unital $C^*$-algebra is~$M_n(\C)$, where the involution operator is the conjugate-transpose and the norm is the operator norm.
A linear map $\pi:\mathcal X\to\mathcal Y$ from one $C^*$-algebra~$\mathcal{X}$ to another~$\mathcal{Y}$ is a \emph{$*$-homomorphism} if it preserves the multiplication operation, $\pi(xy) = \pi(x)\pi(y)$, and satisfies $\pi(x)^* = \pi(x^*)$ for all $x,y\in \mathcal X$.
For a complex Hilbert space~$\HH$, a mapping $\pi:\mathcal X\rightarrow \LH$ is a \emph{$*$-representation} if it is a $*$-homomorphism.
An important fact is the Gelfand--Naimark Theorem~\cite[Theorem~3.4.1]{murphy:2014} asserting that any~$C^*$-algebra admits an isometric (that is, norm-preserving) $*$-representation for some complex Hilbert~space.   Suppose $\X=(X,\cdot_X, *),\Y=(Y,\cdot_Y, \dagger)$ are $C^*$-algebras, then the \emph{tensor product} $\X\otimes \Y$ is also a $C^*$-algebra defined in terms of the standard tensor product of the vector spaces $X\otimes Y$ with the associative multiplication $\cdot _{XY}$ and involution operator $\diamond$ defined as: $(x\otimes y)\cdot_{XY} (x'\otimes y')= (x\cdot_X x')\otimes (y\cdot_Y y')$ and involution $(x\otimes y)^\diamond=x^*\otimes y^\dagger$.  This can then be extended linearly to the entire~domain. 
\medskip

\paragraph{Completely bounded norms}
We also collect a few basic facts about completely bounded norms that we use later and refer to~\cite{Paulsen:2002} for an extensive introduction.
For a $C^*$-algebra~$\mathcal X$ and positive integer~$d$, we denote by~$M_d(\mathcal X)$ the set of $d$-by-$d$ matrices with entries in $\mathcal X$.
Note that this set can naturally be identified with the algebraic tensor product $\mathcal X\otimes L(\C^d)$, that is, the linear span of all elements of the form $x\otimes M$, where $x\in X$ and $M\in L(\C^d)$. Using the Gelfand-Naimark theorem, we endow $M_d(\mathcal X)$ with a norm induced by an isometric $*$-representation~$\pi$ of~$\mathcal X$ into $\LH$ for a Hilbert space~$\HH$. The linear map $\pi\otimes \id_{\ldc}$ sends elements in $M_d(\mathcal X)$ (or~$\mathcal X\otimes \ldc$) to elements (operators) in~$L(\HH\otimes\C^d)$. 
The norm of an element 	$A\in M_d(\mathcal X)$ is then defined to be~$\|A\| = \|(\pi\otimes\id_{\ldc})(A)\|$.
The notation~$\|A\|$ reflects the fact that this norm is in fact independent of the particular $*$-representation~$\pi$.
Based on this, we can define a norm on linear maps $\sigma:\mathcal X\to \LH$ as follows:
$$
\|\sigma\|_{\cb}=\sup \Bigg\{\frac{\|(\sigma\otimes \id_{\ldc})(A)\|}{\|A\|}: d\in \N, \hspace{2 pt} A\in \X\otimes L(\C^d),\hspace{2 pt} A\neq 0 \Bigg\}.
$$

We will also need the following fact about the completely bounded norm of  $*$-representations of $C^*$-algebras \cite[Theorem~1.6]{Pisier:2003}.

\begin{lemma}\label{lem:hom_cb}
	Let~$\mathcal X$ be a finite-dimensional $C^*$-algebra, $\HH, \HH'$ be Hilbert spaces, $\pi:\mathcal X\to L(\HH)$ be a $*$-representation and $U\in L(\HH,\HH')$ and $V\in L(\HH', \HH)$ be linear maps. Then, the map $\sigma:\mathcal X\to L(\HH')$, defined as $\sigma(x) = U\pi(x) V$, satisfies that $\|\sigma\|_{\cb} \leq \|U\|\|V\|$.
\end{lemma}

We will also use the famous Fundamental Factorization Theorem~\cite[Theorem~8.4]{Paulsen:2002}.  Below we state the theorem when restricted to finite-dimensional spaces (see also the remark after~\cite[Theorem~16]{johnston:stabnorms}). 

\begin{theorem}[Fundamental factorization theorem]\label{thm:FFT}
	Let $\sigma:L(\C^n)\to L(\C^m)$ be a linear map and let $d = nm$.
	Then, there exist $U,V\in L(\C^m, \C^{dn})$ such that $\|U\|\|V\| \leq \|\sigma\|_{\cb}$ and for any $M\in L(\C^n)$, we have $\sigma(M) = U^*(M\otimes \id_d)V$.
\end{theorem}

\medskip

\paragraph{Tensors and multilinear forms} 
For vector spaces~$X,Y$ over the same field and positive integer~$t$, recall that a mapping 
\beqn
T:\underbrace{X\times\cdots\times X}_{\text{$t$ times}}\to Y
\eeqn
is \emph{$t$-linear} if for every $x_1,\dots,x_t\in X$ and $i\in[t]$, the map $$y\mapsto T(x_1,\dots,x_{i-1}, y, x_{i+1}, \dots,x_t)$$ is linear. A \emph{$t$-tensor} of dimension $n$ is a map $T:[n]\times \dots \times [n]\rightarrow \C$, which can alternatively be identified by $T=(T_{i_1,\ldots,i_t})_{i_1,\ldots ,i_t=1}^n \in \C^{n\times\dots\times n}$. 
With abuse of notation we identify a $t$-tensor $T\in \C^{n\times\dots\times n}$ with the $t$-linear form $T:\C^n\times\cdots\times\C^n\to\C$ given~by
$$
T(x_1,\dots,x_t)=\sum_{i_1,\ldots,i_t=1}^nT_{i_1,\ldots,i_t} x_{1}(i_1)\cdots x_{t}(i_t).
$$ 

Next, we introduce the general definition of the completely bounded norm of a $t$-linear form $T:\mc X\times\cdots\times\mc X\to\C$ on a $C^*$-algebra~$\mathcal X$.
First, we use the standard identification of such forms with the linear form on the tensor product $\mc X\otimes\cdots\otimes \mc X$ given by $T(x_1\otimes\cdots\otimes x_t)
=
T(x_1,\dots,x_t)$.
We consider a bilinear map $\odot: \big(\mc X\otimes \ldc, \mc X\otimes \ldc\big) \to\mc X\otimes \mc X\otimes\ldc$ for any positive integer~$d$ defined as follows.
For $x,y\in \mc X$ and $M_x,M_y\in \ldc$, let
\beqn
(x\otimes M_x)\odot(y\otimes M_y)
= (x\otimes y)\otimes (M_x M_y).
\eeqn
Observe that this operation changes the order of the tensor factors and \emph{multiplies}~$M_x$ with~$M_y$.
This operation is associative but \emph{not} commutative.
Extend the definition of the $\odot$ operation bi-linearly to its entire domain.
Define the $t$-linear map $T_d:M_d(\mc X)\times\cdots\times M_d(\mc X)\to \ldc$ by
\beq\label{ref:lift}
T_d(A_1\dots,A_t)
=
\big(T\otimes\id_{\ldc}\big)(A_1\odot\cdots\odot A_t).
\eeq
The completely bounded norm of~$T$ is now defined by
\beqn
\|T\|_{\cb}
=
\sup\Big\{
\big\|T_d(A_1,\dots,A_t)\big\|
\st
d\in \N,\:\:
A_j \in M_d(\mc X),\:\:
\|A_j\| \leq 1
\Big\}.
\eeqn 

Note that the definition given in Eq.~\eqref{eq:defnTcb} corresponds to the particular case where the $C^*$-algebra~$\mathcal X$ is formed by the $n\times n$ diagonal matrices.
Since any square matrix with operator norm at most~1 is a convex combination of unitary matrices (by the Russo-Dye Theorem),\footnote{A precise statement and short proof of the Russo-Dye theorem can be found in~\cite{gardnerrussodye:1984}.} the completely bounded norm can also be defined by taking the supremum over unitaries $A_j\in~ M_d(\mc X)$.
The completely bounded norm can be defined even more generally for multilinear maps into~$L(\HH)$, for some Hilbert space~$\HH$, to yield the definition of this norm for linear maps given above, but we will not use this~here.

\paragraph{Quantum query complexity}\label{sec:quantumquerycompl}
The quantum query model was formally defined by Beals et al.~in ~\cite{Beals:2001}.  In this model, we are given black-box access to a unitary operator, often called an oracle~$O_x$, whose description depends in a simple way on some binary input string~$x\in\{0,1\}^n$. 
An application of the oracle on a quantum register is referred to as a quantum \emph{query} to~$x$. 
In the standard form of the model, a query acts on a pair of registers on $(\mathsf Q, \mathsf A)$, where $\mathsf Q$ is an $n$-dimensional query register and $\mathsf A$ is a one-qubit auxiliary register.
A query to the oracle effects the unitary transformation given by
$$
O_x: \ket{i,b}\rightarrow \ket{i,b \oplus x_i}
$$ where $i\in [n]$, $b\in \{0,1\}$. 
(These oracles are also commonly called \emph{bit oracles}.)

A quantum query algorithm consists of a fixed sequence of unitary operations acting on $(\mathsf Q, \mathsf A)$ in addition to a \emph{workspace} register $\mathsf W$. 
A $t$-query quantum algorithm begins by initializing the joint register $(\mathsf Q,\mathsf A,\mathsf W)$ in the all-zero state and continues by interleaving a sequence of unitaries $U_0,\ldots ,U_{t}$ on $(\mathsf Q, \mathsf A, \mathsf W)$ with oracles $O_x$ on $(\mathsf Q, \mathsf A)$. 
Finally, the algorithm performs a $2$-outcome measurement on $\mathsf A$ and returns the measurement outcome.

\begin{figure}[h]
	\begin{center}
		\includegraphics[width = 10cm]{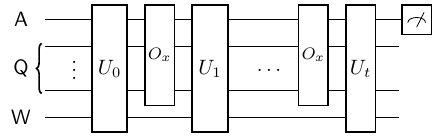}
	\end{center}
	\caption{A $t$-query quantum algorithm that starts with the all-zero state and concludes by measuring the register $\mathsf A$.}
\end{figure}

For a Boolean function $f:\{0,1\}^n\rightarrow \{0,1\}$, the algorithm is said to compute $f$ with error~$\eps > 0$ if for every~$x$, the measurement outcome of register  $\mathsf A$ equals~$f(x)$  with probability at least~$1 - \eps$.
The \emph{bounded-error query complexity} of~$f$, denoted $Q_\eps(f)$, is the smallest~$t$ for which such an algorithm exists. Note that in this model, we are not concerned with the amount of time (i.e., the number of gates) it takes to implement the interlacing unitaries, which could be much bigger than the query complexity~itself. 

Here we will work with a slightly less standard oracle sometime referred to as a \emph{phase oracle}, in which the standard oracle is preceded and followed by a Hadamard on~$\mathsf A$.
Since the Hadamards can be undone by the unitaries surrounding the queries in a quantum query algorithm, using the phase oracle does not reduce generality.
A query to this oracle, sometimes denoted $O_{x,\pm}$, applies the (controlled) unitary $\Diag(({\bf 1}, (-1)^x))$ to joint register~$(\mathsf{A, Q})$.
To avoid having to write $(-1)^x$ later on, we shall work in the equivalent setting where Boolean functions send~$\pmset{n}$ to~$\pmset{}$.

\section{Characterization of quantum query algorithms}\label{sec:mainthmproof}

In this section we prove Theorem~\ref{thm:main}. 
The main ingredient of the proof is the following celebrated representation theorem by Christensen and Sinclair~\cite{christensen:cbrep}  showing that completely-boundedness of a multilinear form is equivalent to the existence of an exceedingly nice factorization. 

\begin{theorem}[Christensen--Sinclair]\label{thm:SC-factor}
	Let $t$ be a positive integer and let $\mc X$ be a $C^*$-algebra. Then, for any $t$-linear form
	$T:\mc X\times\cdots\times \mc X\to \C$, we have $\|T\|_{\cb} \leq 1$ if and only if
	there exist Hilbert spaces $\HH_0,\dots,\HH_{t+1}$ where $\HH_0 = \HH_{t+1} =\C$,  $*$-representations $\pi_i:\mc X\to L(\HH_i)$ for each $i\in[t]$ and contractions $V_i \in L(\HH_{i},\HH_{i-1})$, for each $i\in [t+1]$ 
	such that for any $x_1,\dots,x_t\in \mc X$, we have
	\beq\label{eq:CS-factor}
	T(x_1,\dots,x_t)
	=
	V_1\pi_1(x_1)V_2\pi_2(x_2)V_3\cdots V_{t}\pi_t(x_t)V_{t+1}.
	\eeq
\end{theorem}

We first show how the above result simplifies when
restricting to the special case in which the $C^*$-algebra~$\mc X$  is formed by the set of diagonal $n$-by-$n$ matrices. 

\begin{corollary}\label{cor:CS-diag}
	Let $m,n,t$ be positive integers such that $t \geq 2$ and $m = n^t$. Let $T \in \C^{n\times \cdots\times n}$ be a $t$-tensor.  Then, $\|T\|_{\cb}\leq 1$ if and only if
	there exist a positive integer~$d$, unit vectors $u,v\in \C^{m}$ and contractions $U_i,V_i \in L(\C^m,\C^{dn})$ such that for any $x_1,\dots,x_t\in \C^n$, we have
	\beq\label{eq:mult-factor}
	T(x_1,\dots,x_t)
	=
	u^*U_1^*\big(\Diag(x_1)\otimes \id_{d}\big)V_1 \cdots U_t^*\big(\Diag(x_t)\otimes \id_{d}\big)V_t v.
	\eeq
\end{corollary}

\begin{proof}
	The set $\mc X = \Diag(\C^n)$ of diagonal matrices is a (finite-dimensional) $C^*$-algebra (endowed with the standard matrix product and conjugate-transpose involution).
	Now, define the $t$-linear form $R:\mc X\times\cdots\times \mc X\to\C$ by
	$R(X_1,\dots,X_t) = 
	T(\diag(X_1),\dots,\diag(X_t))$.
	We claim that $\|R\|_{\cb}=\|T\|_{\cb}$.
	Observe that for every positive integer~$d$, the set $\{B\in M_d(\mc X)\st \|B\|\leq 1\}$ can be identified with the set of block-diagonal matrices $B = \sum_{i=1}^nE_{i,i}\otimes B(i)$ of size $nd\times nd$ and blocks $B(1),\dots,B(n)$ of size~$d\times d$ satisfying $\|B(i)\|\leq 1$ for all $i\in[n]$.
	It follows~that
	\begin{align*}
	R_d(B_1,\dots,B_t)
	&=
	\sum_{i_1,\dots,i_t=1}^n R(E_{i_1,i_1},\dots,E_{i_t,i_t})B_1(i_1)\cdots B_t(i_t)\\
	& =
	\sum_{i_1,\dots,i_t=1}^n T_{i_1,\dots,i_t}B_1(i_1)\cdots B_t(i_t),
	\end{align*}
	which shows that $\|R\|_{\cb}=\|T\|_{\cb}$. 
	\vspace{5 pt}
	
	Next, we show that~\eqref{eq:CS-factor} is equivalent to~\eqref{eq:mult-factor}.
	The fact that~\eqref{eq:mult-factor} implies~\eqref{eq:CS-factor} follows immediately from the fact that the map $\Diag(x)\mapsto \Diag(x)\otimes \id_d$ is a $*$-representation.
	Now assume~\eqref{eq:CS-factor}.
	Without loss of generality, we may assume that each of the Hilbert spaces~$\HH_1,\dots,\HH_t$ has dimension at least~$m$.
	If not, we can expand the dimensions of the ranges and domains of the representations~$\pi_i$ and contractions~$V_i$ by dilating with appropriate isometries into larger Hilbert spaces (``padding with zeros'').
	For each $i\in [t]$, let $S_i \subseteq \HH_i$ be the subspace 
	\beqn
	S_i=\Span\big\{\pi_{i}(x_{i})V_{i+1}\cdots V_{t}\pi_t(x_t) V_{t+1}\st x_i,\dots,x_t\in \mc X\big\}.
	\eeqn
	Since $\dim(\mc X) = n$, we have that $\dim(S_i) \leq m$.
	For each $i\in[t]$, let $Q_i \in L(\C^m, \HH_i)$ be an isometry such that $S_i\subseteq \im(Q_i)$.
	Note that~$V_{t+1}$ is a vector in the unit ball of~$\HH_{t}$. Let $Q_{t+1} \in L(\C^{m}, \HH_{t})$ be an isometry such that $V_{t+1}\in \im(Q_{t+1})$. Note that for each $i\in[t+1]$, the map $Q_{i}Q_{i}^*$ acts as the identity on~$\im(Q_i)$. For each $i\in \{2,\dots,t\}$ define the map $\sigma_i:\mc X\to L(\C^m)$ by $\sigma_i(x) = Q_i^*V_{i}\pi_i(x)Q_{i+1}$ and $\sigma_1(x)  = Q_1^* \pi_1(x) Q_2$.
	Finally define $u = Q_1^*V_1^*$ and $v = Q_{t+1}^*V_{t+1}$.
	Then, the right-hand side of~\eqref{eq:CS-factor} can be written as
	\beqn
	u^*\sigma_1(x_1)\cdots \sigma_t(x_t) v.
	\eeqn
	It follows from Lemma~\ref{lem:hom_cb} that $\|\sigma_i\|_{\cb} \leq 1$.
	Let $\sigma_i':L(\C^n)\to L(\C^m)$ be the linear map given by $\sigma_i'(M) = \sigma_i(\Diag(M_{11},\dots,M_{nn}))$ for any $M\in L(\C^m)$.
	Then, for every diagonal matrix~$x\in\mathcal X$, we have $\sigma_i(x) = \sigma_i'(x)$ and also $\|\sigma_i'\|_{\cb} \leq \|\sigma_i\|_{\cb}$. 
	It follows from Theorem~\ref{thm:FFT} that	 there exist a positive integer~$d_i$ and contractions $U_i,V_i:L(\C^m,\C^{dn})$ such that $\sigma'_i(x) = U_i^*(x\otimes\id_{d_i})V_i$ for every $x\in\mc  X$.
	We can take all $d_i$ equal to $d=\max_i\{d_i\}$ by suitably dilating the contractions~$U_i,V_i$.
	Setting $u' = u/\|u\|_{2}$ and $U_1' = \|u\|_{2}U_1$, and similarly defining $v', V_{i+1}'$ shows that  Eq.~\eqref{eq:CS-factor} implies Eq.~\eqref{eq:mult-factor}. 
\end{proof}

Corollary~\ref{cor:CS-diag} implies the following lemma, from which Theorem~\ref{thm:main} easily follows.

\begin{lemma}\label{lem:cb-alg-tensor}
	Let $\beta:\pmset{n}\to [-1,1]$ be some map and let~$t$ be a positive integer.
	Then, the following are equivalent.
	\begin{enumerate}
		\item There exists a $(2t)$-tensor $T\in \R^{2n\times\cdots\times 2n}$ such that $\|T\|_{\cb} \leq 1$ and for every $x\in \pmset{n}$ and $y = (x,{\bf 1})$, we have 
		\beqn
		\sum_{i_1,\dots,i_{2t}=1}^{2n} T_{i_1,\dots,i_{2t}} y_{i_1}\cdots y_{i_{2t}}
		=
		\beta(x).
		\eeqn 
		\item There exists a $t$-query quantum algorithm that, on input $x\in\pmset{n}$, returns a  sign with expected value~$\beta(x)$.
	\end{enumerate}
\end{lemma}

\begin{remark}
	Note that Lemma \ref{lem:cb-alg-tensor} itself already gives a characterization of quantum query algorithms, but in terms of the completely bounded norm of a tensor, as opposed to a polynomial.
	Then, the reader could wonder about the interest of Theorem \ref{thm:main}, which is a similar characterization (though of course, equivalent), but in terms of a more complicated-looking norm. 
	As  mentioned in the introduction, the completely bounded norm of a polynomial can be significantly smaller than that of its associated symmetric tensor.
	Therefore, given a function $\beta$, a symmetric $(2t)$-tensor~$T$ verifying item 1 in Lemma~\ref{lem:cb-alg-tensor} and the degree-$(2t)$ polynomial $p(x) = T(x,\dots,x)$, checking that $\|p\|_{\cb}\leq 1$ should be easier than proving that $\|T\|_{cb}\leq 1$. In fact, it may well be the case that $\|T\|_{cb}>1$ (so Lemma  \ref{lem:cb-alg-tensor} does not allow us to conclude anything, and we should look for another $T$), while $\|p\|_{\cb}\leq 1$ which allows us to apply Theorem \ref{thm:main}.
\end{remark}

\begin{proof}[Proof of Lemma~\ref{lem:cb-alg-tensor}]
	We first prove that $(2)$ implies $(1)$.
	As discussed in Section~\ref{sec:quantumquerycompl}, a $t$-query quantum algorithm with phase oracles  initializes the joint register $(\mathsf{ A,Q, W})$ in the all-zero state on which it then performs some unitaries $U_1,\ldots, U_{t}$ interlaced with queries $D(x) = \Diag(({\bf 1},x))\otimes \id_{\mathsf W}$. Let $\{P_0,P_1\}$ be the the two-outcome measurement done at the end of the algorithm and assume that it returns $+1$ on measurement outcome zero and $-1$ otherwise. 
	Let $Q=P_0-P_1$ and note that $Q$ is a contraction since $P_0,P_1$ are positive semi-definite and satisfy $P_0 + P_1 = \id$.
	The final state of the quantum algorithm (before the measurement of register~$\mathsf A$) is 
	$$
	\psi_x=U_{t} D(x) \cdots U_2 D(x)U_1 e_1.
	$$
	Hence the expected value of the measurement outcome is then given by 
	\beq\label{eq:alg_exp}
	\psi_x^* Q   \psi_x.
	\eeq
	By assumption, this expected value equals $\beta(x)$ for every $x\in \pmset{n}$. For $z\in \C^{2n}$, denote $D'(z)=\Diag((z_{n+1},\dots,z_{2n},z_1,\dots,z_n))\otimes\id_{\mathsf W}$ and $\widetilde{U}_t=U_t^*QU_t$. 
	Define the $(2t)$-linear form~$T$ by
	\begin{align*}
	&T(y_1,\ldots,y_{2t})=
	u^* U_1^*D'(y_1)U_{2}^*  \cdots D'(y_t)\widetilde{U}_t D'(y_{t+1}) \cdots U_2 D'(y_{2t})U_1u.
	\end{align*}
	Clearly $T((x,{\bf 1})\ldots,(x,{\bf 1}))=\beta(x)$ for every $x\in \pmset{n}$.
	Moreover, by definition~$T$ admits a factorization as in~\eqref{eq:mult-factor}. It thus follows from Corollary~\ref{cor:CS-diag} that $\|T\|_{\cb} \leq 1$.
	We turn~$T$ into a real tensor by taking its real part~$T' = (T + \overline{T})/2$, where $\overline{T}$ is the coordinate-wise complex conjugate of~$T$.\footnote{An anonymous referee pointed out one could also use a result of Barnum et al.~\cite{barnum:quantumquerysdp} showing that the unitaries in quantum query algorithms can be assumed to be real. In that case one can assume~$T$ is a real tensor to begin with.}
	Since for any $x\in\pmset{n}$ and $y = (x,{\bf 1})$, the value $T(y,\dots,y)$ is real, we have $T'(y,\dots,y) = \beta(x)$.
	We need to show that $\|T'\|_{\cb} \leq 1$.
	To this end, consider an arbitrary positive integer~$d$, unit vectors $v,w\in \C^d$ and sequences of unitary matrices $V_1(i),\dots,V_{2t}(i)$ for $i\in[n]$ such that
	\begin{align*}
	\Big\|
	\sum_{i_1,\dots,i_{2t}=1}^{2n} \overline{T_{i_1,\dots,i_{2t}}}
	V_1(i_1)\cdots V_{2t}(i_{2t})
	\Big\|
	&=
	\Big|
	\sum_{i_1,\dots,i_{2t}=1}^{2n} \overline{T_{i_1,\dots,i_{2t}}}
	v^*V_1(i_1)\cdots V_{2t}(i_{2t})\, w
	\Big|.
	\end{align*}
	where we assumed that the unit vectors $v,w\in \C^d$ maximize the operator norm. Note that $\|\overline{T}\|_{\cb}$ is given by the supremum over~$d$ and~$V_j(i)$.
	Taking the complex conjugate of the above summands on the right-hand side allows us to express the above absolute value as
	\begin{align}
	\Big|
	\sum_{i_1,\dots,i_{2t}=1}^{2n} T_{i_1,\dots,i_{2t}}
	\bar v^*\overline{V_1(i_1)}\cdots \overline{V_{2t}(i_{2t})}\, \bar w
	\Big|,\label{eq:Tbound}
	\end{align}
	where $\bar{v}, \bar{w}, \overline{V_j(i)}$ denote the coordinate-wise complex conjugates.
	Since each $\overline{V_j(i)}$ is still unitary, it follows that~\eqref{eq:Tbound} is at most $\|T\|_{\cb}$ and so $\|\overline T\|_{\cb} \leq \|T\|_{\cb} \leq 1$.
	Hence, by the triangle inequality, $\|T'\|_{\cb} \leq (\|T\|_{\cb} + \|\overline T\|_{\cb})/2 \leq 1$ as desired.
	\vspace{5 pt}
	
	Next, we show that $(1)$ implies $(2)$. Let~$T$ be a $(2t)$-tensor as in item~1. 
	From Corollary~\ref{cor:CS-diag} it follows that~$T$ admits a factorization as in~\eqref{eq:mult-factor}. Let~$V_0,U_{2t+1} \in L(\C^m, \C^{2dn})$ be isometries.
	For each $i\in [2t+1]$, define the map $W_i\in L(\C^{2dn})$ by $W_i = V_{i-1}U_i^*$.
	Observe that each~$W_i$ is a contraction and
	recall that unitaries are contractions.
	For the moment, assume for simplicity that each~$W_i$ is in fact unitary. Define two vectors $\widetilde u = V_0u$ and $\widetilde v = U_{2t+1}v$ and observe that these are unit vectors in~$\C^{2dn}$.
	The right-hand side of~\eqref{eq:mult-factor} then gives us
	\beq\label{eq:algo_output}
	T(y_1,\ldots,y_{2t})=\widetilde u^* W_1 \widetilde{D}(y_1) W_2\widetilde{D}(y_2)W_3 \cdots W_{2t} \widetilde{D}(y_{2t}) W_{2t+1}\widetilde v,
	\eeq
	where $\widetilde{D}(y_i)=\Diag(y_i)\otimes \id_d$ for $i\in [2t]$. In particular, if we define two unit vectors
	\begin{align*}
	v_1&=(\Diag(y)\otimes \id_d)W_{t}\cdots W_2(\Diag(y)\otimes \id_{d})W_1 \widetilde{u},\\
	v_2&=W_{t+1}^*(\Diag(y)\otimes \id_d)W_{t+2}^*\cdots W_{2t}^*(\Diag(y)\otimes \id_{d}) W_{2t+1}^* \widetilde{v},
	\end{align*}
	then $T(y,\ldots,y)=|v_1^* v_2|$. Based on this, we obtain the quantum query algorithm  that prepares $v_1$ and $v_2$ in parallel, each using at most $t$ queries. This is described in Figure~\ref{fig:t-circuit}.
	
	\begin{center}
		\begin{figure}[!ht]
			\includegraphics[width = 13cm]{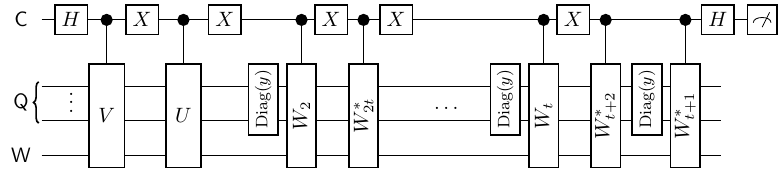}
			\caption{The registers $\mathsf C,\mathsf Q,\mathsf W$ denote the control, query and workspace registers. 
				Let $U,V$ be unitaries with $W_1\widetilde{u}$ and $W_{2t+1}\widetilde{v}$ as their first columns, respectively and for $x\in\pmset{n}$ and $y = (x, {\bf 1})$, let $\Diag(y)$ be the query operator. The algorithm begins by initializing the joint register $(\mathsf C,\mathsf Q,\mathsf W)$ in the all-zero state  and proceeds by performing the displayed operations.
				The algorithm returns~$+1$ if the outcome of the measurement on~$\mathsf C$ equals zero and $-1$ otherwise.}
			\label{fig:t-circuit}
		\end{figure}
	\end{center}
	
	To see why this algorithm satisfies the requirements, first note that the algorithm makes $t$ queries to the input~$x$. For the correctness of the algorithm, we begin by observing that before the application of the first query, the state of the joint register $(\mathsf C,\mathsf Q,\mathsf W)$ is 
	$$
	\tfrac{1}{\sqrt{2}}(e_1\otimes W_1 \widetilde{u}+  e_2\otimes W_{t+1} \widetilde{v}).
	$$
	Before the final Hadamard gate, the state of the joint register  is given by
	\begin{align*}
	&\frac{1}{\sqrt{2}} e_1\otimes\big((\Diag(y)\otimes \id_d)W_{t}\cdots W_2(\Diag(y)\otimes \id_{d})W_1 \widetilde{u}\big) \\+&\frac{1}{\sqrt{2}} e_2\otimes\big(W_{t+1}^*(\Diag(y)\otimes \id_d)W_{t+2}^*\cdots W_{2t}^*(\Diag(y)\otimes \id_{d}) W_{2t+1}^* \widetilde{v}\big).
	\end{align*}
	A standard calculation and~\eqref{eq:algo_output} then show that after the final Hadamard gate, the expected output of the algorithm is precisely $T((x, {\bf 1}),\ldots,(x, {\bf 1})) = \beta(x)$.
	In the general case where the $W_i$s are not necessarily unitary, we can use the fact that, by the Russo--Dye Theorem and Carath\'eodory's Theorem, each~$W_i$ is a convex combination of at most $(dn)^2 + 1$ unitaries.
	The algorithm can thus use randomness to effect each~$W_i$ on expectation.
	Alternatively, by linear algebra there exists a unitary matrix~$W_i'\in \C^{2dn\times 2dn}$ that has~$W_i$ as its upper-left corner (see~\cite[Lemma~7]{Aaronson:2016}), through which the algorithm could implement~$W_i$ by working on a larger quantum register.
\end{proof}

Using Lemma~\ref{lem:cb-alg-tensor}, we now prove our main Theorem~\ref{thm:main}. 

\begin{proof}[Proof of Theorem~\ref{thm:main}]
	We first show that $(2)$ implies $(1)$. Using the equivalence in Lemma~\ref{lem:cb-alg-tensor}, there exists a $(2t)$-tensor $T\in \R^{2n\times\cdots\times 2n}$ such that $\|T\|_{\cb} \leq 1$ and for every $x\in \pmset{n}$ and $y = (x,{\bf 1})$, we have 
	\beqn
	\sum_{i_1,\dots,i_{2t}=1}^{2n} T_{i_1,\dots,i_{2t}} y_{i_1}\cdots y_{i_{2t}}
	=
	\beta(x).
	\eeqn
	Define the symmetric $2t$-tensor $T' = \frac{1}{(2t)!}\sum_{\sigma\in S_{2t}} T\circ \sigma$.
	Let~$p\in \R[x_1,\dots,x_{2n}]$ be the form of degree~$2t$ associated with~$T'$. Since there is a unique symmetric tensor associated with a polynomial, it follows that $T'=T_p$ (where $T_p$ is defined by Eq.~\eqref{eq:Tpdef}). Then,  $p((x,{\bf 1}))=\beta(x)$ for every $x\in \pmset{n}$. 
	Moreover, if we set $T^\sigma = T$ for each $\sigma\in S_t$, it follows  from the above decomposition of~$T_p$ and Definition~\ref{defn:pcb} that $\|p\|_{\cb}\leq \|T\|_{\cb}\leq 1$.

	Next, we show that $(1)$ implies $(2)$. Let $p$ be a degree-$(2t)$ form satisfying $\|p\|_{\cb}\leq~1$. Suppose~$T_p$ as defined in Eq.~\eqref{eq:Tpdef} can be written as $T_p=\sum_{\sigma \in S_{2t}} T^\sigma \circ \sigma$ and $\sum_{\sigma \in S_{2t}} \|T^\sigma\|_{\cb}=\|p\|_{\cb}\leq 1$. Define $T=\sum_{\sigma\in S_{2t}}T^\sigma$. Then, using the triangle inequality, it follows that $\|T\|_{\cb}\leq \sum_{\sigma\in S_{2t}} \|T^\sigma\|_{\cb}\leq 1$. Also note that for any $y\in \R^{2n}$, 
	$$
	T(y,\ldots,y)=\sum_{\sigma \in S_{2t}} T^\sigma (y,\ldots,y)=\sum_{\sigma \in S_{2t}} (T^\sigma  \circ \sigma) (y,\ldots,y)=T_p(y,\ldots,y)=p(y).
	$$
	Using Lemma~\ref{lem:cb-alg-tensor} (in particular $(1)\implies (2)$ in Lemma~\ref{lem:cb-alg-tensor}) for the tensor $T$, the theorem~follows.
\end{proof}
We now prove Corollary~\ref{cor:cbdeg}, which is an immediate consequence of our main theorem.

\begin{proof}[Proof of Corollary~\ref{cor:cbdeg}]
	We first show $\cbdeg_\eps(f)\geq Q_\eps(f)$: Suppose $\cbdeg_\eps(f)=d$, then there exists a degree-$(2d)$ form $p$ satisfying: $|p(x)-f(x)|\leq 2\eps$ for every $x\in D$ and $\|p\|_{\cb}\leq 1$. Using our characterization in Theorem~\ref{thm:main}, it follows that there exists a $d$-query quantum algorithm $\mathcal A$, that on input $x\in D$, returns a  sign with expected value $p(x)$.  So, our $\eps$-error quantum algorithm for $f$ simply runs~$\mathcal A$ and outputs the~sign.
	
	We next show $\cbdeg_\eps(f)\leq Q_\eps(f)$.
	Suppose $Q_\eps(f)=t$. Then, there exists a $t$-query quantum algorithm that, on input $x\in D$, outputs a sign with expected value $\beta(x)$ satisfying $|\beta(x)-f(x)|\leq 2\eps$. Note that we could also run the quantum algorithm for $x\notin D$ and let $\beta(x)$ be the expected value of the quantum algorithm for such $x$s. Using Theorem~\ref{thm:main}, we know that there exists a degree-$(2t)$ form $p$ satisfying $\beta(x)=p(x)$ for every $x\in \pmset{n}$ and $\|p\|_{\cb}\leq 1$. Clearly $p$ satisfies the conditions of Definition~\ref{def:cbdeg}, hence $\cbdeg_\eps(f)\leq t$.
\end{proof}

\section{Separations for quartic polynomials}\label{sec:cubicseparations}
In this section we show the existence of a quartic polynomial~$p$ that is bounded but for which any two-query quantum algorithm~$\mathcal A$ satisfying $\Exp[\mathcal A(x)] = Cp(x)$ for every $x\in\pmset{n}$ must necessarily have $C = O(n^{-1/2})$.
We show this using a (random) \emph{cubic} form that is bounded, but whose  completely bounded norm is $\poly(n)$, following a construction of~\cite[Theorem 18.16]{DJT:1995}.

\begin{remark}
Earlier versions of this work finished this section with a claim that Corollary~\ref{cor:quartic} gives a counterexample to possible quartic extensions of Theorem~\ref{thm:AAIKS}.
Although the claim holds as stated, the proof contained a bug.
The faulty proof is omitted from the current version.
An explanation of the issue, a corrected proof and stronger examples are presented in~\cite{BEG:2022}.
\end{remark} 

Given a form $p: \R^n\rightarrow \R$, we define its norm as $$\|p\|=\sup\{|p(x)|\st x\in \pmset{n}\}.$$
Note that the condition $\|p\|\leq 1$ is equivalent to $p$ being bounded.

\begin{theorem}\label{thm:rand_T}
	There exist absolute constants $C,c\in (0,\infty)$ such that the following holds.  Let\footnote{Recall that $|\alpha|$ in the definition of $p$ is defined as $|\alpha|=\sum_i\alpha_i$.}
	\beqn
	p(x) =\sum_{\alpha\in \{0,1,2,3\}^n:\, |\alpha| = 3} c_\alpha x^\alpha
	\eeqn 
	be a random cubic form such the coefficients~$c_\alpha$ are independent uniformly distributed $\pmset{}$-valued random variables.
	Then, with probability at least $1- Cne^{-c n}$, we have $\|p\|_{\cb} \geq c\sqrt{n}\|p\|$.
\end{theorem}

We shall use the following standard concentration-of-measure results.
The first is the Hoeff\-ding bound~\cite[Corollary 3 (Appendix B)]{pollard:2012}.

\begin{lemma}[Hoeffding bound]\label{lem:hoeffding}
	Let~$X_1,\dots,X_m$ be independent uniformly distributed $\pmset{}$-random variables and let $a\in \R^m$.
	Then, for any $\tau > 0$, we have
	\beqn
	\Pr\Big[\Big|\sum_{i=1}^m a_iX_i\Big| > \tau\big]
	\leq
	2e^{-\frac{\tau^2}{2(a_1^2 + \cdots + a_m^2)}}
	\eeqn
\end{lemma}

The second result is one from random matrix theory concerning upper tail estimates for Wigner ensembles (see~\cite[Corollary~2.3.6]{tao:randommat}).

\begin{lemma}\label{lem:taorandommatrix}  There exist absolute constants $C,c\in (0,\infty)$ such that the following holds. 
	Let $n$ be a positive integer and let~$M$ be a random $n\times n$ symmetric random matrix such that for $j\geq i$, the entries~$M_{ij}$ are independent random variables with mean zero and absolute value at most~1. 
	Then, for any $\tau \geq C$, we have
	\begin{align*}
	\Pr\big[\|M\|>\tau\sqrt{n}\big]\leq Ce^{-c\tau n}.
	\end{align*}
\end{lemma}

We also use the following proposition.

\begin{proposition}\label{prop:commute}
	Let $m,n,t$ be positive integers, let~$p\in\R[x_1,\dots,x_n]$ be a $t$-linear form, let $T_p\in \R^{n\times\cdots\times n}$ be as in~\eqref{eq:Tpdef} and let $A_1,\dots,A_n\in L(\R^m)$ be pairwise commuting contractions.
	Then, 
	\beqn
	\|p\|_{\cb}
	\geq
	\Big\|\sum_{i_1,\dots,i_t=1}^n (T_p)_{i_1,\dots,i_t}A_{i_1}\cdots A_{i_t}\Big\|.
	\eeqn
\end{proposition}

\begin{proof}
	Consider an arbitrary decomposition $T_p=\sum_{\sigma\in S_t}T^{\sigma}\circ \sigma$.
	Then, the definition of the completely bounded norm and triangle inequality show that
	\begin{align*}
	\sum_{\sigma\in S_t}\|T^{\sigma}\|_{\cb}
	&\:\:\geq\:\: \sum_{\sigma\in S_t}\big\|\sum_{i_1,\dots,i_t=1}^nT_{i_1,\dots,i_t}^{\sigma}A_{i_1}\cdots A_{i_t}\big\| \\
	&\:\:\geq\:\: \hspace{-2pt}\big\|\sum_{\sigma\in S_t}\sum_{i_1,\dots,i_t=1}^nT_{i_1,\dots,i_t}^{\sigma}A_{i_1}\cdots A_{i_t} \big\|. 
	\end{align*}
	Since the $A_i$ commute, the above can be re-written as
	\begin{align*}
	\big\|\sum_{\sigma\in S_t}\sum_{i_1,\dots,i_t=1}^n T_{i_1,\dots,i_t}^{\sigma} A_{i_{\sigma^{-1}(1)}}\cdots A_{i_{\sigma^{-1}(t)}}\big\|
	&= \big\|\sum_{\substack{i_1,\dots,i_t\in [n]\\ \sigma \in S_t}}(T^{\sigma}\circ\sigma)_{i_1,\dots,i_t}A_{i_1} \cdots A_{i_t} \big\| \\
	&=\big\|\sum_{i_1,\dots,i_t=1}^n(T_p)_{i_1,\dots,i_t}A_{i_1}\cdots A_{i_t}\big\|.
	\end{align*}
	The claim now follows from the definition of~$\|p\|_{\cb}$ and using the fact that the decomposition of~$T_p$ was arbitrary.
\end{proof}

\begin{proof}[Proof of Theorem~\ref{thm:rand_T}]
	We begin by showing that with high probability, $\|p\| \leq O(n^2)$. 
	To this end, let us fix an arbitrary $x\in \pmset{n}$. 
	Then, $p(x)$ is a sum of at most $n^3$ independent uniformly distributed random $\pmset{}$-random variables.
	It therefore follows from Lemma~\ref{lem:hoeffding} that 
	\beqn
	\Pr\big[|p(x)| > 2n^2\big] \leq 2e^{-2n},
	\eeqn
	By the union bound over $x\in \pmset{n}$, it follows that $\|p\| > 2n^2$ with probability at most $2e^{-n}$, which gives the claim.

	We now lower bound $\|p\|_{\cb}$. 
	Let~$\tau> 0$ be a parameter to be set later.
	Let~$T\in \R^{n\times n\times n}$ be the random symmetric 3-tensor associated with~$p$ as in~\eqref{eq:Tpdef}.
	For every $i\in [n]$, we define the linear map $A_i:\R^{2n+2}\to\R^{2n+2}$ by
	$$\left\{
	\begin{array}{l}
	A_ie_1= e_i\\
	A_ie_j= \frac{1}{\tau\sqrt{n}}\sum_{k=1}^nT_{i,j,k}e_{k+n}\\ 
	A_ie_{j+n}= \delta_{i,j}e_{2n+1}\\
	A_ie_{2n+1}= e_1.
	\end{array}
	\right.$$
	Observe that for every $i,j,k\in[n]$, we have
	\beq\label{eq:AAA}
	e^*_{2n+1}A_{i}A_{j}A_{k}e_1 = \frac{1}{\tau\sqrt{n}} T_{i,j,k}.
	\eeq
	Since~$T$ is symmetric, it follows easily that these maps commute, which is to say that $A_iA_j=A_jA_i$ for every $i,j\in [n]$. 
	In addition, we claim that with high probability, these maps are contractions (i.e., the associated matrices have operator norm at most~1).
	To see this, for each $i\in[n]$, let $M_i$ be the random matrix given by $M_i= (T_{i,j,k})_{j,k=1}^n$.
	Observe that $M_i$ is symmetric and its entries have mean zero and absolute value at most~1.
	By Lemma~\ref{lem:taorandommatrix} and a union bound, we get that 
	\begin{align}\label{operator norm random matrix}
	\Pr\Big[\max_{i\in [n]}\big\|M_i\big\|> \tau\sqrt{n}\Big]\leq Cne^{-c\tau n}.
	\end{align}
	for absolute constants $c,C$ and provided $\tau\geq C$.
	Now, for any Euclidean unit vector $u\in \R^{2n+2}$, we~have
	\begin{align*}
	\|A_iu \|^2 &=|u_0|^2+\frac{1}{\tau^2 n}\sum_{k=1}^n\Big|\sum_{j=1}^nu_jT_{i,j,k}\Big|^2+|u_{i+n}|^2\\
	&\leq 
	|u_0|^2+\frac{\|M_i\|^2}{\tau^2n}\sum_{j=1}^n|u_j|^2+|u_{i+n}|^2.
	\end{align*}
	It follows from~\eqref{operator norm random matrix} that $\max_i\|M_i\|\leq \tau\sqrt{n}$ with probability at least $1-Cne^{-c\tau n}$, which in turn  implies the above is at most~$\|u\|^2 \leq 1$ and therefore that all~$A_i$ have operator norm at most~1.
	
	By Proposition~\ref{prop:commute},
	\beqn
	\|p\|_{\cb}
	\geq
	\big\|\sum_{i,j,k=1}^nT_{i, j, k}A_{i}A_{j}A_{k}\big\|,
	\eeqn
	provided that the~$A_i$s are contractions.
	
	By~\eqref{eq:AAA}, and since $|T_{i,j,k}| \geq 1/6$ for every $i,j,k\in[n]$, the above is at least $n^{5/2}/(36\tau)$.
	with probability at least $1 - Cne^{-c\tau n}$. 
	Letting~$\tau$ be a sufficiently large constant then gives the result.
\end{proof}

As mentioned in the introduction, one can easily extend this result to the case of $4$-linear forms.  To demonstrate the failure of Theorem~\ref{thm:AAIKS} for quartic polynomials, we ``embed" the degree-$3$ polynomial $p$ in Theorem~\ref{thm:rand_T} into a degree-$4$ polynomial $q$ which has high completely bounded norm.

\begin{corollary}\label{cor:quartic}
	There exists a bounded quartic form
	\beq\label{eq:quartic}
	q(x_1,\dots,x_n) = \sum_{\alpha \in \bset{n} \st |\alpha| = 4}d_\alpha x^\alpha,
	\eeq
	and pairwise commuting contractions $A_1,\dots,A_n \in L(\R^{2n+2})$ such that
	\beqn
	\Big\|\sum_{i,j,k,\ell=1}^n(T_q)_{i,j,k,\ell}A_iA_jA_kA_\ell\Big\| \geq c\sqrt{n}
	\eeqn
	where $c\in (0,1]$ is some absolute constant.
\end{corollary}

\begin{proof}
	Let $p$ be a bounded multi-linear cubic form such that $\|p\|_{\cb} \geq C\sqrt{n}$, the existence of which is guaranteed by Theorem~\ref{thm:rand_T}.
	Let $T_p\in \R^{n\times n\times n}$ be the random symmetric $3$-tensor associated to~$p$. Consider the symmetric 4-tensor $S\in \R^{(n+1)\times (n+1)\times (n+1)\times (n+1)}$ defined by $S_{0,j,k,\ell}=T_{j,k,\ell}$, $S_{i,0,k,\ell}=T_{i,k,\ell}$, $S_{i,j,0,\ell}=T_{i,j,\ell}$, $S_{i,j,k,0}=T_{i,j,k}$ for every $i,j,k,\ell \in [n]$ and $S_{i,j,k,\ell}=0$ otherwise. Since $S$ is symmetric, there exists a unique multi-linear quartic form~$q$ associated to~$S$. It follows easily that $\|q\|= 4\|p\|$. Moreover, by considering the contractions $A_i$ used in the  proof of Theorem~\ref{thm:rand_T} and defining $A_0=\id_{n+2}$, it follows that~$\|q\|_{\cb}\geq 4\|p\|_{\cb}$.
	The form $q/4$ is thus as desired.
\end{proof}

\section{Short proof of Theorem~\ref{thm:AAIKS}}\label{sec:AAIKS-proof}
In this section, we give a short proof of Theorem~\ref{thm:AAIKS}.

\medskip

\paragraph{Proof sketch of Theorem~\ref{thm:AAIKS}} 
We begin by giving a brief sketch of the original proof. The first step is to show that without loss of generality, we may assume that the polynomial~$p$ is a quadratic form.
This is the content of the decoupling argument mentioned in the introduction, proved for polynomials of arbitrary degree in~\cite{Aaronson:2016}, but stated here only for the quadratic case.

\begin{lemma}\label{lem:decoupling}
	There exists an absolute constant $C \in (0,1]$ such that the following holds.
	For any bounded quadratic polynomial $p$, there exists a matrix $A\in \R^{(n+1)\times (n+1)}$ with $\|A\|_{\ell_\infty\to \ell_1} \leq 1$, such that the quadratic form $q(y) = y^{\mathsf T}Ay$ satisfies $q((x,1)) = Cp(x)$ for all $x\in\pmset{n}$.
\end{lemma}

To prove the theorem, we may thus restrict to a quadratic form~$p(x) = x^{\mathsf T}Ax$ given by some matrix~$A\in \R^{n\times n}$ such that $\|A\|_{\ell_\infty \to \ell_1}\leq 1$.
The next step is to massage the matrix~$A$ into a unitary matrix (that can be applied by a quantum algorithm).
To obtain this unitary, the authors use an argument based on two versions of Grothendieck's inequality and a technique known as \emph{variable splitting}, developed in earlier work of Aaronson and Ambainis~\cite{Aaronson:2015}.
The first version of Grothendieck's inequality is the one most commonly used in applications~\cite{Grothendieck:1953}.

\begin{theorem}[Grothendieck] \label{thm:grothendieck}
	There exists a universal constant $K_G \in (0,\infty)$ such that the following holds. 
	For every positive integer~$n$ and matrix~$A\in \R^{n\times n}$, we~have
	\beqn
	\sup\Big\{
	\sum_{i,j=1}^nA_{ij}\langle u_i,v_j\rangle
	\st
	d\in\N,\:
	u_i,v_j\in B_2^n
	\Big\}
	\leq
	K_G\, \|A\|_{\ell_\infty\to\ell_1}.
	\eeqn
\end{theorem}

Elementary proofs of this theorem can be found in~\cite{Alon:2006}. The \emph{Grothendieck constant}~$K_G$ is the smallest real number for which Theorem~\ref{thm:grothendieck} holds true.
The problem of determining its exact value, posed in~\cite{Grothendieck:1953}, remains open.
The best lower and upper bounds $1.6769  \cdots\leq K_G < 1.7822 \cdots$ were proved by Davie and Reeds~\cite{Davie:1984, Reeds:1991}, and Braverman et al.~\cite{Braverman:2013}, resp.
The second version of Grothendieck's inequality is as~follows.

\begin{theorem}[Grothendieck]\label{thm:Grothfact}
	For every positive integer~$n$ and matrix~$A\in \R^{n\times n}$, there exist $u,v\in (0,1]^n$ such that $\|u\|_{2} = \|v\|_{2} = 1$ and such that the matrix
	\beq\label{eq:Grothfact}
	B  = \frac{1}{K_G}\, \Diag(u)^{-1}A\Diag(v)^{-1}
	\eeq
	satisfies
	$\|B\| \leq \|A\|_{\ell_\infty\to\ell_1}$,
	where $\Diag(w)$ denotes the square diagonal matrix whose diagonal is~$w$.
\end{theorem}
\medskip

\paragraph{Our contribution} 
The first (standard) version of Grothendieck's inequality (Theorem~\ref{thm:grothendieck}) easily implies that any matrix~$A$ such that $\|A\|_{\ell_\infty\to\ell_1} \leq 1$ has completely bounded norm at most~$K_G$. Combing this fact with our Theorem~\ref{thm:main} and Lemma~\ref{lem:decoupling}, one quickly retrieves Theorem~\ref{thm:AAIKS}.
However, Theorem~\ref{thm:main} is based on the rather deep Theorem~\ref{thm:SC-factor}.
We observe that Theorem~\ref{thm:AAIKS} also follows readily from the much simpler Theorem~\ref{thm:Grothfact} alone (proved below for completeness), after one assumes that~$p$ is a quadratic form as above.

Indeed, Theorem~\ref{thm:Grothfact} gives unit vectors~$u,v$ such that the matrix~$B$ as in~\eqref{eq:Grothfact} has (operator) norm at most~1.
Unitary matrices have norm exactly~1 and of course represent the type of operation a quantum algorithm can implement.
Moreover, since $u,v$ are unit vectors, they represent $(\log n)$-qubit quantum states.
Using the fact that for $w,z\in \R^n$, we have $\Diag(w)z = \Diag(z)w$, we get the following  \emph{factorization} formula (not unlike the one of Corollary~\ref{cor:CS-diag}, which is of course no coincidence):
\beq\label{eq:factor}
\frac{x^{\mathsf T}Ax}{K_G} 
=
x^{\mathsf T}\Diag(u) B\Diag(v) x
=
u^{ T}\Diag(x) B\Diag(x) v.
\eeq
If we assume for the moment that the matrix~$B$ actually is unitary,
then the right-hand side of~\eqref{eq:factor} suggests the simple one-query quantum algorithm described in Figure~\ref{fig:circuit1}.

\begin{figure}[!h]
	\begin{center}
		\includegraphics[width = 10cm]{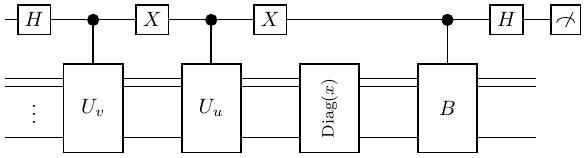}
	\end{center}
	\caption{Let $U_u,U_v$ be unitaries that have $u,v$ as their first columns, respectively.
		The algorithm initializes a $(1 + \log n)$-qubit register in the all-zero state,
		transforms this state into the superposition $\tfrac{1}{\sqrt{2}}(e_1	\otimes u+  e_2\otimes v)$,
		queries the input~$x$ via the unitary~$\Diag(x)$ applied to the $(\log n)$-qubit register, applies a controlled-$B$, and
		finishes by
		measuring the first qubit in the Hadamard basis.
	}\label{fig:circuit1}
\end{figure}

Using~\eqref{eq:factor}, we observe that the algorithm returns zero with probability

\begin{align*}
\frac{1}{2} + \frac{1}{2} \big\langle \Diag(u)x, B\Diag(v)x\big\rangle =
\frac 12+ \frac{x^{\mathsf T}Ax}{2K_G},
\end{align*}

Now, it is clear that the the expected value of the measurement result is precisely $p(x)/K_G$, giving Theorem~\ref{thm:AAIKS} with $C = 1/K_G$.
In case~$B$ is not unitary, one can use the same argument as in the final step of the proof of Theorem~\ref{thm:main}.

\subsection{Factorization version of Grothendieck's inequality}
For completeness and because of its relevance to Theorem~\ref{thm:AAIKS}, we here give a proof of Theorem~\ref{thm:Grothfact}. 
The proof relies on the standard version of Grothendieck's inequality (Theorem~\ref{thm:grothendieck}).
In addition, the proof makes use of the following version of the Hahn--Banach theorem~\cite[Theorem 3.4]{rudin:1991}. % (see Tao's blog~\cite{taoblog:hahnbanach}, Theorem~4~there).

\begin{theorem}[Hahn--Banach separation theorem]\label{thm:hbseparation}
Let $C,D\subseteq \R^n$ be convex sets and let~$C$ be algebraically open.
Then the following are equivalent:
\begin{itemize}
\item The sets~$C$ and~$D$ are disjoint.
\item There exists a vector $\lambda\in\R^n$ and a constant~$\alpha\in \R$ such that $\langle \lambda, c\rangle < \alpha$ for every $c\in C$ and $\langle \lambda, d\rangle \geq \alpha$ for every $d\in D$.
\end{itemize}
Morever, if~$C$ and~$D$ are convex cones\index{Convex cone},\footnote{A convex cone $\mathcal K$ is a set that satisfies: $(i)$ for every $x\in \mathcal K$ and $\lambda>0$, we have $\lambda x\in \mathcal K$ and $(ii)$ for every $x,y \in \mathcal K$, we have $x+y\in \mathcal K$.} we may take~$\alpha = 0$.
\end{theorem}

\begin{proof}[Proof of Theorem~\ref{thm:Grothfact}]
Let $M=A/(K_G \|A\|_{\ell_\infty\rightarrow \ell_1})$.  By  Theorem~\ref{thm:grothendieck} (the standard Grothendieck inequality), we have that

\beqn
\sum_{i,j=1}^nM_{ij}\langle x_i,y_j\rangle
\leq 1
\eeqn
for all vectors~$x_i,y_j$ with Euclidean norm at most~1.
Then, for arbitrary vectors $x_i,y_j$, we have
\beq\label{eq:normalizingandamgm}
\sum_{i,j=1}^nM_{ij}\langle x_i,y_j\rangle
\leq \max_{i,j\in[n]}\|x_i\|\|y_j\|
\leq
\frac{1}{2}\max_{i,j\in[n]}(\|x_i\|^2 + \|y_j\|^2),
\eeq
where the second inequality is by AM-GM inequality. Define the set $K\subseteq \R^{n\times n}$ by
\beqn
K
=
\left\{
\Big(
\|x_i\|^2 + \|y_j\|^2
-
2\sum_{k,\ell=1}^nM_{k\ell}\langle x_k,y_\ell\rangle
\Big)_{i,j=1}^n
\st
d\in\N,\:\:
x_i,y_j\in \R^d
\right\}.
\eeqn

We claim that~$K$ is a convex cone. Observe that for every $t\in \R_+$ and matrix $Q\in K$ given by vectors $x_i,y_j$, the vectors $x'_i = \sqrt{t}x_i$ and $y_j' = \sqrt{t}y_j$ similarly define~$tQ$, and so~$K$ is a cone.
We now show that $K$ is a convex set. 
Let $Q, Q'\in K$ be specified by $x_i,y_j$ and $x_i', y_j'$ respectively.  Then, for any $\lambda\in [0,1]$, the convex combination $\lambda Q + (1 - \lambda)Q'$ also belongs to~$K$, as it can be specified by the  vectors $(\sqrt{\lambda} x_i, \sqrt{1 - \lambda}x_i'),(\sqrt{\lambda} y_j, \sqrt{1 - \lambda}y_j')$.

	Additionally, it follows from Eq.~\eqref{eq:normalizingandamgm} that~$K$ is disjoint from the open convex cone $\R_{<0}^{n\times n}$ of matrices with strictly negative entries. By Theorem~\ref{thm:hbseparation} (the Hahn--Banach separation theorem), we conclude that there exists a nonzero matrix~$L\in \R^{n\times n}$ such that $\langle L, Q\rangle \geq0$ for every $Q\in K$ and $\langle L, N\rangle < 0$ for every $N\in \R_{<0}^{n\times n}$.
In particular, the second inequality implies that $L\in \R_+^{n\times n}$. Let $P=L/\sum_{ij} L_{ij}$, so that $\{P_{ij}\}_{i,j=1}^n$ defines a probability distribution over $[n]^2$.
Then, for any $Q\in K$,
\begin{align*}
0 &\leq \langle P,Q\rangle\\
&= \sum_{i,j=1}^n P_{ij}(\|x_i\|^2 + \|y_j\|^2) - 2\sum_{k,\ell=1}^nM_{k\ell}\langle x_k,y_\ell\rangle\\
&= \sum_{i=1}^n\sigma_i\|x_i\|^2 + \sum_{j=1}^n\mu_j\|y_j\|^2 - 2\sum_{k,\ell=1}^nM_{k\ell}\langle x_k,y_\ell\rangle,
\end{align*}
where $\sigma_i = P_{i1} + \cdots + P_{in}$ and $\mu_j = P_{1j} + \cdots + P_{nj}$. Observe that $\sigma_i,\mu_j$ are strictly positive because $P_{ij}>0$.
Rearranging the inequality above and using bi-linearity, it follows that for every $\lambda >0$, we~have
\begin{align}
\label{eq:lambdaub}
2\sum_{k,\ell=1}^nM_{k\ell}\langle x_k,y_\ell \rangle
&=
2\sum_{k,\ell=1}^nM_{k\ell}\langle \lambda x_k,\lambda^{-1}y_\ell \rangle \notag\\
&\leq
\lambda^2\sum_{i=1}^n\sigma_i\|x_i\|_2^2 + \lambda^{-2}\sum_{j=1}^n\mu_j\|y_j\|_2^2.
\end{align}
Setting
\beqn
\lambda = 
\left(
\frac{\sum_{j=1}^n\mu_j\|y_j\|_2^2}{\sum_{i=1}^n\sigma_i\|x_i\|_2^2}
\right)^{1/4}
\eeqn
in Eq.~\eqref{eq:lambdaub}, we find that
\beqn
2\sum_{k,\ell=1}^nM_{k\ell}\langle x_k,y_\ell \rangle
\leq
\Big(\sum_{i=1}^n\sigma_i\|x_i\|_2^2\Big)^{1/2}
\Big(\sum_{j=1}^n\mu_j\|y_j\|_2^2 \Big)^{1/2}.
\eeqn
In particular, for the case where $x_k,y_\ell\in \R$, i.e., the scalar case, we have
\beqn
x^{\mathsf T}My \leq \|\diag(\sigma)^{1/2}x\|_2 \|\diag(\mu)^{1/2}y\|_2.
\eeqn
This implies 
$$
x^{\mathsf T} \Big( \Diag(\sigma)^{-1/2}M\Diag(\mu)^{-1/2}\Big) y\leq \|x\|_2\cdot \|y\|_2,
$$
which in particular implies that $\|\Diag(\sigma)^{-1/2}M\Diag(\mu)^{-1/2}\|\leq 1$. Using the definition of $M=A/(K_G \|A\|_{\ell_\infty\rightarrow \ell_1})$, we~have
$$
\|\Diag(\sigma)^{-1/2}A\Diag(\mu)^{-1/2}\|\leq K_G \|A\|_{\ell_\infty\rightarrow \ell_1}.
$$
The theorem follows by letting $u_i=\sqrt{\sigma_i}$, $v_i=\sqrt{\mu_i}$ for every $i\in[n]$.
\end{proof}

\section*{Acknowledgments}
J.B.\ thanks Farrokh Labib for useful discussions. S.A. thanks Tom Bannink, Nathaniel Johnston for useful discussions and Ronald de Wolf for helpful discussions and encouragement. 
We also thank the anonymous referees from QIP, ITCS, and SICOMP for their helpful comments.
We thank Andris Ambainis and Martins Kokainis for pointing out a minor inconsequential error in an earlier version of this paper regarding the normalization of the tensor~\eqref{eq:Tpdef}.

\bibliographystyle{alpha}
\bibliography{qalg_poly}

\end{document}